\documentclass{article}
\usepackage[utf8]{inputenc}
\usepackage[T1]{fontenc}
\usepackage[light]{kpfonts}
\usepackage[utf8]{inputenc}
\usepackage[english]{babel}
\usepackage{amsmath}
\usepackage{cellspace}

\usepackage[hidelinks]{hyperref}
\usepackage[a4paper, total={6.5in, 9.3in}]{geometry}
\usepackage{graphicx}
\usepackage[shortlabels]{enumitem}
\usepackage{comment} 
\usepackage{amsthm}
\usepackage{setspace}
\usepackage{xcolor}

\theoremstyle{definition}

\newtheorem{example}{Example}[section]
\newtheorem{theorem}{Theorem}[section]

\newtheorem{proposition}{Proposition}[section]
\newtheorem{remark}{Remark}[section]
\begin{document}
\title{\textbf{Classification of degenerate non-homogeneous Hamiltonian operators}}
\author{ Marta Dell'Atti$^1$ \qquad  Pierandrea Vergallo$^{2,3}$ \\[5mm]
 \small $^1$  School of Mathematics \& Physics \\ 
\small University of Portsmouth \\
\small \texttt{marta.dell'atti@port.ac.uk}\\
\small $^2$  Department of Mathematical, Computer,
\small  Physical and Earth Sciences\\
\small  University of Messina, \\
\small   V.le F. Stagno D'Alcontres 31, I-98166 Messina, Italy\\
\small  \texttt{pierandrea.vergallo@unime.it} \\
\small $^3$ {Istituto Nazionale di Fisica Nucleare, Sez.\ Lecce} 
}

\date{\small}

\maketitle 

\begin{abstract}
We investigate non-homogeneous Hamiltonian operators composed of a first order Dubrovin-Novikov operator and an ultralocal one. The study of such operators turns out to be fundamental for the inverted system of equations associated with a class of Hamiltonian scalar equations. Often, the involved operators are degenerate in the first order term.  For this reason a complete classification of the operators with degenerate leading coefficient in systems with two and three components is presented.
\end{abstract}

\maketitle
\section{Introduction}
%{\color{magenta} intro dell intro}

\noindent
The Hamiltonian formalism for Partial Differential Equations (PDEs) is one of the leading tools to study nonlinear systems \cite{NovikovManakovPitaevskiiZakharov:TS,DubKri}, following the well known developed theory for finite dimensional ones. As shown in \cite{KerstenKrasilshchikVerbovetsky:HOpC,KVV17}, Hamiltonian operators link conserved quantities with symmetries of the system, mapping the former onto the latter, and then leading to a deeper investigation of the structure of the solutions. This formalism represents a strong theoretical and practical connection between geometry and mathematical physics
\cite{mokhov98:_sympl_poiss,DubrovinNovikov:PBHT,dorfman91:_local}.  Dubrovin and Novikov introduced differential-geometric Poisson brackets as a natural extension of finite dimensional symplectic structures in traditional Hamiltonian mechanics that turned out to arise in several examples in nonlinear PDEs. The characterisation of these structures is related to (pseudo-) Riemannian geometry and algebraic geometry, especially for systems in $1+1$ dimensions (or in independent variables $x,t$).

More in general, geometrical methods are well established tools widely used to find solutions to systems, such as the generalized hodograph method introduced by  Tsarev \cite{tsarev91:_hamil}, valid for strictly hyperbolic systems. The Hamiltonian formalism is also used to discuss the integrability. In particular, 
finding two compatible Hamiltonian structures is strictly related to the existence of infinitely many commuting symmetries and conservation laws, as proved by Magri~\cite{Magri:SMInHEq}. In the context of hydrodynamic type systems, an approach to describe integrability is based on the analysis of geometrical elements in the so called method of hydrodynamic reductions, developed for systems in $2+1$ dimensions \cite{FerKar,Mor} and then extended to systems in $1+1$ dimensions with infinitely many components~\cite{FeraMarshall}.

First, we recall some basic notions concerning the Hamiltonian formalism \cite{mokhov98:_sympl_poiss}. Let us consider a system described by $n$ field variables
\begin{equation}  
u^i=u^i(t,x)\,, \qquad  i = 1, \dots, n 
\end{equation} 
depending on the independent variables $t$, $x$, and let $u^i_\sigma$ denote the $x$-derivatives of $u$ $\sigma$ times. A Hamiltonian operator is a linear operator $ A^{ij} = a^{ij \sigma} D_{\sigma}$ such that the associated bracket for functionals $f,g$ \begin{equation}
    \{f,g\}=\int{\frac{\delta f}{\delta u^i}\,A^{ij}\,\frac{\delta g}{\delta u^j}\,dx} \,,
\end{equation} is a Poisson bracket, i.e.\ it is bilinear, skew-symmetric and satisfies the Jacobi identity
\begin{equation*}
 \begin{split}
 & \{f,g\}=-\{g,f\}\,, \hspace{2ex}  \\[1.5ex]
 & \{f,\{g,h\}\}+\{g,\{h,f\}\}+\{h,\{f,g\}\}=0 \,.
 \end{split} 
\end{equation*}
An evolutionary system 
\begin{equation}\label{1}
    u^i_t=F^i(x,{u},{u}_\sigma) \,,
\end{equation}
with $u = \{u^i\}_{i=1}^n$ is Hamiltonian if it admits the following representation
\begin{equation}
    u^i_t=F^i(x,{u},{u}_\sigma)=A^{ij} \,\frac{\delta H}{\delta u^j} \,,
\end{equation} 
where $\delta$ is the variational derivative, $H$ is the Hamiltonian functional 
$
    H = \int h(u) \, dx \,, 
$
written in terms of the Hamiltonian density $h$, and $A$ is a Hamiltonian operator. 

In \cite{DN83,DubrovinNovikov:PBHT}, Dubrovin and Novikov present  a class of Hamiltonian operators which are homogeneous in the order of derivation and are also known as  \emph{homogeneous Hamiltonian operators}. They prove \cite{DN83} that first order homogeneous operators of the form  
\begin{equation}
    g^{ij}\partial_x+\Gamma^{ij}_ku^k_x \,, 
    \label{eq:dub_nov_op}
\end{equation} 
with $\det{g}\neq 0$, are Hamiltonian if and only if
\begin{equation}
\begin{split}
         g_{ij}&=(g^{ij})^{-1} \,, \qquad \Gamma^{ij}_k =-g^{is}\,\Gamma^j_{sk},
\end{split}
\end{equation}
i.e.\ $g^{ij}$ is a flat metric and the coefficients $\Gamma^{i}_{jk}$ are Christoffel symbols for the the metric tensor $g$. Operators of this type naturally arise in homogeneous quasilinear systems of first order PDEs, also known as \emph{hydrodynamic type systems}
\begin{equation}
    u^i_t = v^i_j(u) \, u^j_x  \,, \qquad i = 1, \dots, n\,,
    \label{eq:hydro_syst}
\end{equation} where  $v(u)=( v^{i}_j )_{1 \le i,j\le n}$ is the coefficient matrix depending on the field variables. Indeed,  if~\eqref{eq:hydro_syst} is Hamiltonian with a Dubrovin--Novikov operator \eqref{eq:dub_nov_op}, it can be expressed as 
\begin{align}\begin{split}
    u^i_t&= v^i_j(u)\,u^j_x=A^{ij}\,\frac{\partial h}{\partial u^j}\\ &=\left(g^{ij}\partial_x+\Gamma^{ij}_ku^k_x\right)\, \frac{\partial h}{\partial u^j} = \left(\nabla^i\,\nabla_j\, h \right) \, u^j_x \,,
    \label{eq:quas}
\end{split}\end{align}
where $h=h(u)$ is the hydrodynamic Hamiltonian density and $\nabla_i$ the covariant derivative. 

In the case of higher order homogeneous operators of degree $m$, the Dubrovin--Novikov operator generalizes as 
\begin{align}\label{hhom}
   \begin{split}
A^{ij}=&g^{ij}D^m_x+b^{ij}_k\,u^k_x\,D^{m-1}_x+\left(C^{\,ij}_k\,u^k_{xx}+C^{\,ij}_{kl}\,u^k_x\,u^l_x\right)D^{m-2}_x\\
&\hspace{2ex}+ \dots  \,+\left(d^{ij}_k\,u^k_{nx}+\dots +d^{ij}_{k_1\dots k_m}\,u^{k_1}_x\,\cdots\, u^{k_m}_x\right) \,,
\end{split}
\end{align}  
where the coefficients $b^{ij}_k,\, C^{\,ij}_k,\, \dots $ depend on the field variables.
%In \cite{FPV14,FerPavVit1} E.V.Ferapontov, M.V.Pavlov and R.Vitolo prove that third order HHOs possess interesting projective geometric properties  and they provide a projective classification of the operators up to 4 components {\color{magenta}(stessa cosa di sempre, e' una dimensione dello spazio o del num di componenti?)}. Following the methods introduced for third order structures, in \cite{VerVit1,VerVit2} P.Vergallo and R.Vitolo prove their projective invariance and find a complete classification for HHOs of second order.  
Dubrovin and Novikov present also an extension of homogeneous structures\cite{DubrovinNovikov:PBHT}, introducing \emph{non-homogeneous} Hamiltonian operators as sum of two or more homogeneous ones. A leading example in this context is offered by the Korteweg-De Vries equation
\begin{equation}
    u_t=6uu_x+u_{xxx} \,,
\end{equation}
which possesses a Hamiltonian structure through the operator
\begin{equation}
    A=\partial_x^{3}+2u\partial_x+u_x \,,
\end{equation} 
given by the sum of the third order operator $\partial_x^3$ and the first order operator $2u\partial_x+u_x$. 

Following the notation used by Dubrovin and Novikov, if an operator is given by the sum of two homogeneous operators of order $k$ and $m$ respectively, we denote the order of the non-homogeneous operator via the sum $k+m$.  

Let us consider the simplest case, with $k=1$ and $m=0$ for the so-called \emph{non-homogeneous operators of hydrodynamic type} $1+0$. They naturally arise in non-homogeneous quasilinear systems of first-order PDEs. Let $C^{\,ij}=A^{ij}+\omega^{ij}$, where $A^{ij}$ is homogeneous of order 1 and $\omega^{ij}$ is a symplectic structure of order 0. One can easily generalise \eqref{eq:quas} to systems of this type
\begin{align}
\begin{split}
u^i_t=&\left(g^{ij}\partial_x+\Gamma^{ij}_ku^k_x+\omega^{ij}\right)\,\frac{\partial h}{\partial u^j}\\=&\left( \nabla^i\,\nabla_j\, h\right) \, u^j_x + \tilde{\nabla}^i h \,,
\end{split}
\end{align}
where $\tilde{\nabla}^i=\omega^{is}\, \nabla_s$ and ${\nabla}_s$ is the standard gradient.

A remarkable example of a non-homogeneous quasilinear system possessing such a construction is given by the $3$-waves equation \cite{mokhov98:_sympl_poiss} \begin{equation}\begin{cases}\label{3wav}
u^1_t=-c_1u^1_x-2(c_2-c_3)u^2u^3\\[1.2ex]
u^2_t=-c_2u^2_x-2(c_1-c_3)u^1u^3\\[1.2ex]
u^3_t=-c_3u^3_x-2(c_2-c_1)u^1u^2\end{cases} \,,
\end{equation}
that is a Hamiltonian with operator
\begin{equation}C^{\,ij}=
\begin{pmatrix}
1&0&0\\0&-1&0\\0&0&-1
\end{pmatrix}\partial_x+\begin{pmatrix}
0&-2u^3&2u^2\\
2u^3&0&2u^1\\-2u^2&-2u^1&0
\end{pmatrix} \,.
\end{equation}
Finally, following the approach introduced by S. I. Tsarev \cite{tsa3}, we will see how non-homogeneous hydrodynamic operators arise in a class of systems  obtained by the inversion of an evolutionary  Hamiltonian equation (see Theorem \ref{139}). Often, the inversion of such equations leads to a degeneration of the leading coefficient $g^{ij}$ in the first order operator. This is a strong motivation for the investigation of degenerate $1+0$ structures.

In this paper, we present a complete classification of Hamiltonian operators for systems in two and three components of the form 
\begin{equation}
    C^{\,ij}=g^{ij}\partial_x+b^{ij}_ku^k_x+\omega^{ij}\,,
\end{equation}
focusing on the case when the leading coefficient is degenerate (i.e.\ its rank is lower than the number of components of the system) and some related remarkable examples of systems in $1+1$ dimensions exhibiting this feature. The importance of a deeper study of such operators has been remarked by O. I. Mokhov in \cite{mokhov98:_sympl_poiss}, who finds Hamiltonian structures of this type in the study of the real reduction of 2-waves interaction system, but also by Dubrovin and Novikon themselves \cite{DubrovinNovikov:PBHT}. 

In section \ref{2}, we introduce the conditions for non-homogeneous operators of hydrodynamic type to be Hamiltonian, either with non-degenerate or degenerate assumptions. We establish the connection between such operators and non-homogeneous systems of first order PDEs, introducing the corresponding inverted systems and their associated Hamiltonian structures. 
In section \ref{3}, we show a complete classification, up to diffeomorphisms of the manifold defined by the field variables, of degenerate operators of type $1+0$ for systems with two and three components. In section \ref{applications}, we provide several examples with Hamiltonian structures fitting the above mentioned classification, with particular emphasis on inverted Hamiltonian systems. 

\section{Non-homogeneous hydrodynamic operators}\label{2}
In this section we review non-homogeneous operators of hydrodynamic type, as originally introduced in \cite{DubrovinNovikov:PBHT} and further investigated in~\cite{FerMok1,koinho}.

Non-homogeneous operators of hydrodynamic type are introduced as the natural generalization of  homogeneous Hamiltonian operators \eqref{hhom}
\begin{equation}
  C^{\,ij}=  g^{ij}\partial_x+b^{ij}_ku^k_x+\omega^{ij} \,,
  \label{eq:operator_c}
\end{equation}
where $g^{ij},b^{ij}_k$ and $\omega^{ij}$ depend on the field variables ${u}$. Then, the underlying non-homogeneous local Poisson structure of hydrodynamic type is defined as
\begin{equation}
\begin{split}
    &\{\, u^i(x),u^j(y) \,\} = g^{i j}(u(x)) \, \delta_x (x-y) \\
    &\hspace{8ex}  + b^{i j}_{k} (u(x))\, u^k_x\, \delta (x-y)  + \omega^{i j}(u(x)) \, \delta (x-y) \,.
    \end{split}
\end{equation}
Notice that operators of type $1+0$ are composed of two homogeneous operators $A^{ij}=g^{ij}\partial_x+b^{ij}_ku^k_x$ of order $1$, and $\omega^{ij}$ of order $0$.  
The conditions for $C^{\,i j}$ to be Hamiltonian are given by the constraints for each of its homogeneous part to be Hamiltonian (Theorem  \ref{th:ham_A} and Theorem \ref{th:ham_omega} respectively) and an additional compatibility condition between the two (Theorem \ref{thm1}). 

We recall that operators of order zero, also known as \emph{ultralocal operators}, are Hamiltonian if the following conditions are satisfied.

\begin{theorem}\hspace*{-1ex}\cite{mokhov98:_sympl_poiss} \label{th:ham_omega}
The operator $\omega^{ij}(u)$ is Hamiltonian if and only if it forms  a finite-dimensional Poisson structure, i.e.\ it satisfies the conditions
\begin{align}
        &\omega^{i j} (u)= - \omega^{j i}(u)  \,,  \label{cond3} \\
        &\omega^{i s} \frac{\partial \omega^{j k}}{\partial u^s} + \omega^{j s} \frac{\partial \omega^{k i}}{\partial u^s} + \omega^{k s} \frac{\partial \omega^{i j}}{\partial u^s} =0 \,. \label{cond4} 
\end{align}
 
\end{theorem}

We remark that in the non-degenerate case, i.e.\ $\det \omega^{ij}\neq 0$, conditions \eqref{cond3} and \eqref{cond4} are respectively skew-symmetry and closedness of the 2-form $\omega$.

In the case of operators of first order the following result holds. 
\begin{theorem}\hspace*{-1ex}\cite{mokhov98:_sympl_poiss} \label{th:ham_A}
The operator $A^{ij}$ is Hamiltonian if and only if
\begin{align} 
&g^{ij}=g^{ji} \\
&\frac{\partial g^{ij}}{\partial u^k} = b^{ij}_k+b^{ji}_k \\
&g^{is}b^{jk}_s-g^{js}b^{ik}_s=0 \\
&g^{is}\left(\frac{\partial b^{jr}_{s}}{\partial u^k}-\frac{\partial b^{jr}_{k}}{\partial u^s}\right)+b^{ij}_sb^{sr}_k-b^{ir}_s b^{sj}_k =0  \\
&g^{is}\frac{\partial b^{jr}_q}{\partial u^s}-b^{ij}_sb^{sr}_q-b^{ir}_sb^{js}_q = g^{js}\frac{\partial b^{ir}_q}{\partial u^s}-b^{ji}_sb^{sr}_q-b^{is}_qb^{jr}_s \,, \\
&\sum_{(q,k)} \left\{ \frac{\partial}{\partial u^q}\left(g^{is}\left(\frac{\partial b^{jr}_{s}}{\partial u^k}-\frac{\partial b^{jr}_{k}}{\partial u^s}\right)+b^{ij}_sb^{sr}_k-b^{ir}_sb^{sj}_k\right)  +\displaystyle \sum_{(i,j,k)} \left(b^{si}_q\left(\frac{\partial b^{jr}_k}{\partial u^s}-\frac{\partial b^{jr}_s}{\partial u^k}\right)\right) \right\} =0
\end{align} 
with the sum over $(q,k)$ and $(i,j,k)$ is on cyclic permutations of the indices. 
\end{theorem}

Let us remark that here there is no assumption about the non-degeneracy properties of metric. The conditions for non-homogeneous operators of hydrodynamic type to be Hamiltonian are shown in the following theorem.
\begin{theorem}\hspace*{-1ex}\cite{FerMok1, mokhov98:_sympl_poiss}\label{thm1}
The operator \eqref{eq:operator_c} is Hamiltonian if and only if $g^{i j}\, \partial_x + b^{i j}_{\,\, k} \, u^k_x$ 
is Hamiltonian, $\omega^{i j}$ is Hamiltonian, and the compatibility conditions are satisfied
\begin{align}\label{cond1}
    \Phi^{i j k} &= \Phi^{k i j} \,,
\\\label{eq:phi}
        \frac{\partial \Phi^{i j k}}{\partial u^r} & = \sum_{(i,j,k)}  b^{s i}_{r}\, \frac{\partial \omega^{j k}}{\partial u^s} + \left( \frac{\partial b^{i j}_{r}}{\partial u^s} - \frac{\partial b^{i j}_{s}}{\partial u^r}  \right)\omega^{s k}\,,
\end{align}
where $\Phi^{i j k}$ is the $(3,0)$-tensor
\begin{align}
        \Phi^{i j k} & = g^{i s}\, \frac{\partial \omega^{j k}}{\partial u^s} - b^{i j}_{s} \, \omega^{s k} - b^{i k}_{s} \, \omega^{j s} \,. 
\end{align} 
\end{theorem}

\section{Classification for systems in two and three components}\label{3}
Savoldi\cite{sav1} presents a complete classification of degenerate first order homogeneous operators for systems with two and three components. Starting from these results, in this section we provide a novel complete classification of degenerate operators of type 1+0. 
To obtain an explicit form of~$\omega^{ij}$ by means of Theorem \ref{thm1} it is sufficient to solve conditions~\eqref{cond1} and~\eqref{eq:phi} with fixed tensors~$g^{ij}$ and~$b^{ij}_k$, giving rise to an overdetermined system of PDEs. %, which can be solved by the help of symbolic computations. 
In addition, we require the ultralocal operator $\omega^{ij}$ to be Hamiltonian imposing \eqref{cond3} and \eqref{cond4}, via Theorem \ref{th:ham_omega}. { In appendix A we report the details of the computations.} 

The following computations are carried out with the support of computer algebra methods, implemented in Maple, Reduce and Mathematica. 
The use of symbolic computation for integrable systems and Hamiltonian structures is itself an ongoing topic of research\cite{KVV17,Vit1}.

%\vspace*{-5ex}

\subsection{Systems in \texorpdfstring{$n=2$}{n=2} components}
Let us consider systems with two components, with field variables $u$, $v$. In general, given $n$ the number of components of the hydrodynamic system, in the degenerate case the operator $g^{ij}$ can be classified by its rank, with $\textup{rank}\left( g^{i j}\right)~=~m<n$. In the following, we explicit the number of components $n$ for the operator $C^{\,ij}_{n,k}$ while the index $k$ is used to distinguish between different operators. 

For $n=2$, $\textup{rank} (g^{i j})$ is in $\{0,1\}$. The only solution for the case $\textup{rank} (g^{i j})=0 $ is the trivial one, then the operator reduces to a symplectic form. 
In the case $\textup{rank} (g^{i j})=1 $, we can construct two different operators,    
    \begin{align}
        C^{\,i j}_{2,1} &= \begin{pmatrix} 
        \partial_x & 0 \\ 
        0 & 0 
        \end{pmatrix} + \begin{pmatrix}
        0 & f( v ) \\
        -f( v ) & 0 
        \end{pmatrix} \,, \label{eq:C2_1} \\[1.5ex] 
        C^{\,i j}_{2,2} &= \begin{pmatrix} 
        \partial_x & 0 \\[1.5ex]  
        0 & 0 
        \end{pmatrix} + \begin{pmatrix}
        0 & -\dfrac{v_x}{u} \\ 
        \dfrac{v_x}{u} & 0 
        \end{pmatrix} + \begin{pmatrix}
        0 & \dfrac{f( v )}{u} \\[1.5ex] 
        -\dfrac{f( v )}{u} & 0 
        \end{pmatrix} \,,
    \end{align}
where $f(v)$ is an arbitrary function depending only on the variable $v$.

\begin{theorem}\label{thmde2}
Every degenerate operator of type $1+0$ in two components can be mapped either onto an ultralocal Hamiltonian operator or onto one between $C^{\,ij}_{2,1}$ and $C^{\,ij}_{2,2}$.
\end{theorem}
\begin{proof}
Considering Theorem \ref{thm1}, we compute the symplectic structure satisfying \eqref{cond1} and \eqref{eq:phi} for each degenerate operator of the classification introduced by Savoldi in two components. 
\end{proof}

%\vspace*{-6ex}

\subsection{Systems in \texorpdfstring{$n=3$}{n=3} components}
Let us consider the case of systems with three components $u$, $v$, $w$, for which the degenerate metric has $\textup{rank}\left( g^{i j} \right)$ in~$\{0,1,2\}$. We denote with $f,g,h,l$ arbitrary functions, specifying the explicit dependence on the variables, and with $c$ arbitrary constants.   

\begin{itemize}
    \item $\textup{rank} \left( g^{i j} \right) =0 $
    %\begin{equation}
    %    C_{3,0}^{ij}=\omega^{ij}
    %\end{equation}
    %{\color{orange} Chiusura su questa }
    \begin{equation}
        C_{3,1}^{i j} = 
        \begin{pmatrix}
        0 & w_x & 0 \\
        -w_x & 0 & 0 \\
        0 & 0 & 0 \\
        \end{pmatrix} + 
        \begin{pmatrix}
        0 & f(u,v,w) & 0 \\
        -f(u,v,w) & 0 & 0 \\
        0 & 0 & 0 \\
        \end{pmatrix} \,,
    \end{equation}
    
    \item $\textup{rank} \left( g^{i j} \right) =1 $
    \begin{equation}
        C_{3,2}^{i j} = 
        \begin{pmatrix}
        \partial_x & 0 & 0 \\
        0 & 0 & 0 \\
        0 & 0 & 0 \\
        \end{pmatrix} + 
        \begin{pmatrix}
        0 & f(v,w) & g(v,w) \\
        -f(v,w) & 0 & h(v,w) \\
        -g(v,w) & -h(v,w) & 0 \\
        \end{pmatrix} \,,
        \label{eq:C2_three}
    \end{equation}
    where the function $f(v,w)$ is expressed in terms of the functions $h(v,w)$ and $g(v,w)$ as 
    \begin{equation}
        f(v,w) = h(v,w) \left( l(w)\,  +  \int_{1}^v \frac{g(s,w) \partial_w h(s,w) - h(s,w) \partial_w g(s,w)}{h(s,w)^2} \, ds \right) \,.
        \label{eq:funcional_expr_C2_C5}
    \end{equation}

    \begin{align}
        C_{3,3}^{i j} &= 
        \begin{pmatrix}
        \partial_x & 0 & 0 \\
        0 & 0 & 0 \\
        0 & 0 & 0 \\
        \end{pmatrix} + 
        \begin{pmatrix}
        0 & w_x & 0 \\
        -w_x & 0 & 0 \\
        0 & 0 & 0 \\
        \end{pmatrix} + 
        \begin{pmatrix}
        0 & f(v,w) & 0 \\
        -f(v,w) & 0 & 0 \\
        0 & 0 & 0 \\
        \end{pmatrix} \,, \\[2ex]
   % \end{equation} 
   %  \begin{equation}
        C_{3,4}^{i j} &= 
        \begin{pmatrix}
        \partial_x & 0 & 0 \\[1ex]
        0 & 0 & 0 \\[1ex]
        0 & 0 & 0 \\
        \end{pmatrix} + 
        \begin{pmatrix}
        0 & 0 & -\dfrac{w_x}{u} \\
        0 & 0 & 0 \\
        \dfrac{w_x}{u} & 0 & 0 \\
        \end{pmatrix} + 
        \begin{pmatrix}
        0 & 0 & \dfrac{f(v,w)}{u} \\[1ex]
        0 & 0 & 0 \\[1ex]
        -\dfrac{f(v,w)}{u} & 0 & 0 \\
        \end{pmatrix} \,, \\[2ex]
  %  \end{equation} 
   %   \begin{equation}
        C_{3,5}^{i j} &= 
        \begin{pmatrix}
        \partial_x & 0 & 0 \\[1ex]
        0 & 0 & 0 \\[1ex]
        0 & 0 & 0 \\
        \end{pmatrix} + 
        \begin{pmatrix}
        0 & -\dfrac{v_x}{u} & -\dfrac{w_x}{u} \\
        \dfrac{v_x}{u} & 0 & 0 \\[1ex]
        \dfrac{w_x}{u} & 0 & 0 \\
        \end{pmatrix} + 
        \begin{pmatrix}
        0 & \dfrac{f(v,w)}{u} & \dfrac{g(v,w)}{u} \\[1ex]
        -\dfrac{f(v,w)}{u} & 0 & \dfrac{h(v,w)}{u} \\[1ex]
        -\dfrac{g(v,w)}{u} & -\dfrac{h(v,w)}{u} & 0 \\
        \end{pmatrix} \,,
   % \end{equation} 
   \end{align}
    with $f(v,w)$ given in \eqref{eq:funcional_expr_C2_C5}.

\vspace{5ex}

    \item $\textup{rank} \left( g^{i j} \right) =2 $
 \begin{align}
        C_{3,6}^{i j} &= 
        \begin{pmatrix}
        \partial_x & 0 & 0 \\[1ex]
        0 & \partial_x & 0 \\[1ex]
        0 & 0 & 0 \\
        \end{pmatrix} + 
        \begin{pmatrix}
        0 & f(w) & g(w) \\[.5ex]
        -f(w) & 0 &c\,g(w) \\
        -g(w) & -c\, g(w) & 0 \\
        \end{pmatrix} \,, \\[3ex]
        C_{3,7}^{i j} &= 
        \begin{pmatrix}
        \partial_x & 0 & 0 \\[1ex]
        0 & \partial_x & 0 \\[1ex]
        0 & 0 & 0 \\
        \end{pmatrix} + 
        \begin{pmatrix}
        0 & 0 & 0 \\
        0 & 0 & -\dfrac{w_x}{v} \\
        0 & \dfrac{w_x}{v} & 0 \\
        \end{pmatrix} + 
        \begin{pmatrix}
        0 & 0 & c f(w) \\[1ex]
        0 & 0 & \dfrac{(1-c u)f(w)}{v} \\[1ex]
        -c f(w) & -\dfrac{(1-c u)f(w)}{v} & 0 \\
        \end{pmatrix} \,,\\[3ex]\end{align}
        \begin{align}
\begin{split} 
        C_{3,8}^{i j} &= 
        \begin{pmatrix}
        \partial_x & 0 & 0 \\[2ex]
        0 & \partial_x & 0 \\[2ex]
        0 & 0 & 0 \\
        \end{pmatrix} \,+ \,
        \begin{pmatrix}
        0 & 0 & -\dfrac{w w_x}{u w -v} \\[1.5ex]
        0 & 0 & \dfrac{w_x}{u w -v} \\[1.5ex]
        \dfrac{w w_x}{u w -v} & -\dfrac{w_x}{u w -v} & 0 \\
        \end{pmatrix} \\[1.5ex] \,& \hspace{3ex} + (1+w^2) \, f(w) 
        \begin{pmatrix}
        0 & \dfrac{1}{(1+w^2)} & \dfrac{\left( w - c\, v \sqrt{1+w^2} \right)}{uw-v} \\[2.5ex]
        -\dfrac{1}{(1+w^2)} & 0 & -\dfrac{\left( 1 - c\, u \sqrt{1+w^2}  \right) }{uw -v} \\[2.5ex]
        -\dfrac{\left( w -  c\, v \sqrt{1+w^2}  \right)}{uw-v} & \dfrac{\left( 1 - c \, u \sqrt{1+w^2}  \right) }{uw -v} & 0 \\
        \end{pmatrix} \,,
      \end{split}   
    \end{align}  
    \begin{align}
        C_{3,9}^{i j} &= 
        \begin{pmatrix}
        0 & \partial_x & 0 \\
        \partial_x & 0 & 0 \\
        0 & 0 & 0 \\
        \end{pmatrix} + 
        \begin{pmatrix}
        0 & f(w) & c\, g(w) \\
        -f(w) & 0 & g(w) \\
        -c\,g(w) & -g(w) & 0 \\
        \end{pmatrix}  \,, \\[2ex]
     \begin{split} 
        C_{3,10}^{i j} &= 
        \begin{pmatrix}
        0 & \partial_x & 0 \\[1ex]
        \partial_x & 0 & 0 \\[1ex]
        0 & 0 & 0 \\
        \end{pmatrix} + 
        \begin{pmatrix}
        0 & 0 & -\dfrac{w_x}{v} \\[1ex]
        0 & 0 & 0 \\[1ex]
        \dfrac{w_x}{v} & 0 & 0 \\
        \end{pmatrix} +  %\\ & \hspace{5ex} + 
        \begin{pmatrix}
        0 & f(w) & \dfrac{h(w) - u g(w)}{v}  \\[1ex]
        -f(w) & 0 & g(w) \\[1.5ex]
        -\dfrac{h(w) - u g(w)}{v} & -g(w) & 0 \\[1.5ex] 
        \end{pmatrix} \,,
        \end{split} 
    \end{align}
    with the additional condition 
    \begin{equation}\label{41}
        h(w) g'(w) - g(w) \left( f(w) + h'(w) \right) =0 \,,  
    \end{equation}
\begin{equation}
    \begin{split}
             C_{3,11}^{i j} &= 
        \begin{pmatrix}
        0 & \partial_x & 0 \\[2ex]
        \partial_x & 0 & 0 \\[2ex]
        0 & 0 & 0 \\
        \end{pmatrix} + 
        \begin{pmatrix}
        0 & 0 & \dfrac{w_x}{u w -v} \\[1.5ex]
        0 & 0 & -\dfrac{w w_x}{u w -v} \\[1.5ex]
        -\dfrac{w_x}{u w -v} & \dfrac{w w_x}{u w -v} & 0 \\
        \end{pmatrix} \\[1.5ex] \,& \hspace{7ex} 
        + f(w)
        \begin{pmatrix}
        0 & \dfrac{c }{\sqrt{w}} & \dfrac{\left( uw - 2c \sqrt{w}  \right) }{uw -v}  \\[1ex]
        -\dfrac{c }{\sqrt{w}} & 0 & -\dfrac{w\left( v - 2 c \sqrt{w}\right) }{uw - v}  \\[1ex]
        -\dfrac{\left( uw - 2c \sqrt{w}  \right) }{uw -v} & w\, \dfrac{w\left( v - 2 c \sqrt{w}\right) }{uw - v} & 0 \\
        \end{pmatrix}
    \end{split}
\end{equation}
\end{itemize}
\begin{remark}
Condition \eqref{41} can be explicitly solved with respect to any function among~$f$, $g$ and $h$. 
\end{remark}

\begin{theorem}\label{thmde3}
Every degenerate operator of type $1+0$ in three components can be mapped either onto an ultralocal operator satisfying the closure relation, or onto one among $C^{\,ij}_{3,k}$ with~$k=1,\dots , 11$. 
\end{theorem}
\begin{proof}
Imposing the conditions on the operators to be Hamiltonian, we obtain the extension of the classification for degenerate first order operators presented by Savoldi in three components~\cite{sav1}. { See Appendix A for more details. }
\end{proof}

\begin{remark}
In the proposed classification, we have considered three arbitrary functions for the sake of generality and in view of possible relevance for applications. However, we emphasise that changes of variables can simplify the form of operators. To do so one should look for those changes of variables that leave the order $1$ operator invariant, then apply them to the order $0$ one. 
\end{remark}

\section{Applications}\label{applications}
In this section we present some examples of non-homogeneous quasilinear systems with degenerate Hamiltonian structure of order $1+0$ in two and three components. 

\begin{example}[$2$-wave interaction system]
Mokhov\cite{mokhov98:_sympl_poiss} studies the real reduction of 2-waves interaction system formulated in terms of the system of hydrodynamic equations in two field variables $u=u(x,t)$ and $v=v(x,t)$
\begin{equation}\label{mok1}
\begin{cases}
u_t=a uv \\
v_t=a v_x+u^2
\end{cases}\,,
\end{equation}
with $a$ constant. The system admits a Hamiltonian formulation, with the operator
\begin{equation}
C^{\,i j}=\begin{pmatrix}
0&0\\0&\partial_x
\end{pmatrix}\,+\,\begin{pmatrix}
0&-u\\u&0
\end{pmatrix} \,,
\label{eq:ham_op_2waves}
\end{equation} 
and the Hamiltonian functional
\begin{equation}
    H=\frac{1}{2}\int{\left(a v^2-u^2\right)\, dx} \,.
\end{equation}
The $1+0$ operator \eqref{eq:ham_op_2waves} is degenerate, since the rank of the order $1$ term is lower than the number of components of the system. Moreover, by applying the exchange $u \leftrightarrow v$ it is evident that the operator found by Mokhov is of type $C_{2,1}^{ij}$ in \eqref{eq:C2_1}.
\end{example}

\begin{example}[Sinh-Gordon equation]
Let us consider the Sinh-Gordon equation
\begin{equation}
    \varphi_{\tau\xi}=\sinh{\varphi} \,.
\end{equation}
Applying the change of variables $\varphi=2\log u$, we have 
\begin{equation}
    \left(2\, \frac{u_\tau}{u}\right)_\xi=\dfrac{1}{2}\left(u^2-\dfrac{1}{u^2}\right) \, .
\end{equation} 
Introducing $v={2u_\tau}/{u}$ and considering the light-cone coordinates $\tau=t, \xi=t-x$
\begin{equation}
    \begin{cases}
    u_t=\dfrac{1}{2}uv\\
    v_t=v_x+\dfrac{1}{2}\left(u^2-\dfrac{1}{u^2}\right)
    \end{cases},
\end{equation}
we show that the system is Hamiltonian with the non-homogeneous hydrodynamic operator of shape \eqref{eq:C2_1} with the exchange of variables $u \leftrightarrow v$ and the function $f(u)=u/2$
\begin{equation}
    C^{\,i j}=\begin{pmatrix}
0&0\\0&\partial_x
\end{pmatrix}\,+\frac{1}{2}\begin{pmatrix}
0&u\\-u&0
\end{pmatrix} \,.\end{equation} 
The corresponding Hamiltonian density is 
\begin{equation}
    h(u,v)=\dfrac{1}{2}\left(v^2-u^2+\dfrac{1}{u^2}\right) \,.
\end{equation}
\end{example}

\subsection{Inverted Hamiltonian systems}\label{2.1}
In this section we show the connection between degenerate operators of type $1+0$ and scalar equations possessing a local Hamiltonian structure. 

Let us briefly recall that the momentum of a Hamiltonian equation $u_t=A^{ij} \, {\delta H}/{\delta u^j}$ is a functional defining a $x$-translation
\begin{equation*}
    u_x^i=A^{ij} \, \frac{\delta P}{\delta u^j}\,, \qquad \text{ with } P=\int p(u,u_\sigma)\, dx \,,
   \label{mom}
\end{equation*}
for  $i = 1, \dots, n$. Tsarev\cite{tsa3} proves that under the inversion of the independent variables $x$ and $t$, the Hamiltonian property is preserved by the system. It is well known that the momentum is a conserved quantity in a Hamiltonian system, hence there exists $q(u,u_\sigma)$ such that $p_t=q_x$. Then, one can choose $H'=\int{q(u,u_\sigma)\, d t}$ as the Hamiltonian functional of the inverted system.

Non-homogeneous operators of hydrodynamic type are related to the study of scalar evolutionary equations
possessing a \emph{local} Hamiltonian  structure. Indeed, by introducing the new set of variables
\begin{equation}
u^1=u\,,\hspace{1ex} u^2=u_x\,,\hspace{1ex} u^3=u_{xx}\,, \hspace{1ex} \dots  \,,
\label{eq:change_var}
\end{equation}
it is in some cases possible to write an equivalent non-homogeneous quasilinear system that can be seen as evolutionary with respect to the independent variable $x$, obtaining the \emph{inverted system}.

\begin{remark}Let us observe that every invertible system of order $k$ has the form
\begin{equation}\label{284}
    u_t=F_1(u,u_x,\dots , u_{(k-1)x})+F_2(u,u_x,\dots , u_{(k-1)x}) u_{kx} \,,
\end{equation}
where $F_1,F_2$ are arbitrary functions. Note that this is the case of KdV and many other examples in nonlinear phenomena. Indeed, considering the lower derivatives as parameters, we need the system to be linear in $u_{kx}$ in order to conserve linearity in $u_t$ once inverted.\end{remark}

The following result offers an explicit connection between non-homogeneous hydrodynamic operators and inverted systems.
\begin{proposition}\label{139} Let us consider the evolutionary equation $u_t = F(u,u_{\sigma})$ endowed with a local Hamiltonian structure and a momentum density $p$ in \eqref{mom} depending on  $u$ only. Then, if the inverted system in the set of variables  \eqref{eq:change_var} admits a local Hamiltonian structure, this is given in terms of a non-homogeneous operator of hydrodynamic type.
\end{proposition}
\begin{proof}We observe the following 
\begin{equation}
    q_x=p_t=p_u(u)\, u_t=p_u(u)\, F(u,u_\sigma) \,, \hspace{2ex} \sigma \leq k \,, 
\end{equation}
where $p_t$ is of order $\le k$, at most equal to the order of the equation, and so is $q_x$. Hence, $q(u,u_\sigma)$ is of order at most $k-1$. This implies that the Hamiltonian $H'= \int q(u^1,\dots, u^{k-1}) \, dt$ is of hydrodynamic type for the inverted system in the new  variables. 
In \cite{tsa3,hampds}, it has been proved that the Hamiltonian property is preserved after a change of dependent variables and an inversion of $t$ and $x$. Then, the inverted system is quasilinear of first order and already Hamiltonian. %Moreover, in \cite{hampds} Kersten, Verbovtski and Vitolo prove that the \textbf{locality} property of the operator is also preserved. 
The operator $B^{\,ij}$ in 
\begin{equation}
    u^i_x = B^{\,ij} \, \dfrac{\delta H'}{\delta u^j} \,,
\end{equation}
being local, must be of type $1+0$, i.e.\ a non-homogeneous operator of hydrodynamic type. 
\end{proof}

The Proposition \ref{139} %strongly 
justifies a deeper investigation of such operators, for which KdV offers a leading example, as follows. %Unfortunately, 
We emphasise the previous theorem does not guarantee that the operator is in general non-degenerate. 

\vspace{2ex}

\begin{example}[KdV equation - I]\label{ex0}
Let us consider the KdV equation
\begin{equation}
    u_t=6uu_x+u_{xxx} \,,
\end{equation} which is widely known to be Hamiltonian. Inverting the equation, we obtain the evolutionary system with respect to $x$ in three components $u^1(x,t)$, $u^2(x,t)$,  $u^3(x,t)$ defined as $u = u^1 , \,u_x = u^2,\, u_{xx} = u^3 \,,$
yielding the following non-homogeneous system of hydrodynamic type
\begin{equation}\label{kdvsys}
\begin{cases}
u^1_x=u^2\\[1.2ex]
u^2_x=u^3\\[1.2ex]
u^3_x=u^1_t+6u^1u^2
\end{cases}\,.
\end{equation} 
This system is Hamiltonian with the following non-homogeneous hydrodynamic type operator \cite{tsa3}
\begin{equation}
    C^{\,ij}=\begin{pmatrix}
    0&0&0\\0&0&0\\0&0&1
    \end{pmatrix}\partial_t+\begin{pmatrix}
    0&1&0\\-1&0&6u^1\\0&-6u^1&0
    \end{pmatrix} \,,
    \end{equation}
   with the leading coefficient $g^{ij}$ being degenerate.     It is easy to see that applying the change of variables $u^1=w$, we obtain again the operator \eqref{eq:C2_three}, where %$g(v,w)=0$,$f(v,w)=-6w$, $l(w)=6w$ and  $h(v,w)=-1$.
   
 \vspace{1ex}  
   \centering{
    \begin{tabular}{l l}
       $g(v,w)=0\,,$  & $\hspace{2ex} f(v,w)=h(v,w)\,l(w)\,,$ \\[1.5ex]
        $l(w)=6w\,,$ & $\hspace{2ex} h(v,w)=-1\,.$
    \end{tabular}
}
\end{example}

\vspace{1ex}

\begin{example}[KdV equation - II]\label{ex1}
Mokhov \cite{mk2} finds a  transformation of variables (also known as \emph{local quadratic unimodular change})
\begin{equation}\label{change}
\begin{split}
u^1=\frac{w^1-w^3}{\sqrt{2}}\,,\qquad u^2=w^2\,,\qquad  
u^3=\frac{w^1+w^3}{\sqrt{2}}+\left(w^1-w^3\right)^2 \,,
\end{split} 
\end{equation}
such that the KdV equation reads as 
\begin{equation}\label{kdvsys2}
    \begin{cases}
    w^1_x=-\dfrac{1}{2}\left(w^1-w^3\right)_t+w^2\left(w^1-w^3\right)+\dfrac{1}{\sqrt{2}}w^2\\
    w^2_x=\left(w^1-w^3\right)^2+\dfrac{1}{\sqrt{2}}\left(w^1+w^3\right)\\
    w^3_x=-\dfrac{1}{2}\left(w^1-w^3\right)_t+w^2\left(w^1-w^3\right)-\dfrac{1}{\sqrt{2}}w^2    \end{cases} \,.
\end{equation}
After this local change, the KdV is a bi-Hamiltonian system with respect to two non-homogeneous operators $1+0$ of hydrodynamic type, one of these being the operator 
\begin{equation}\label{1368}
\begin{split} 
&C^{\,ij}=\frac{1}{2}\begin{pmatrix}
1&0&1\\0&0&0\\1&0&1
\end{pmatrix}\partial_t +\begin{pmatrix}
0&w^1-w^3+\frac{1}{\sqrt{2}}&0\\
w^3-w^1-\frac{1}{\sqrt{2}}&0&w^3-w^1+\frac{1}{\sqrt{2}}\\
0&w^1-w^3-\frac{1}{\sqrt{2}}&0
\end{pmatrix} \,,
\end{split} 
\end{equation}
which is degenerate, since $\text{rank}(g^{i j})=1$. 
The Hamiltonian given in terms of the new variables is 
\begin{equation}
    H=\int{\left((w^1)^2-(w^2)^2-(w^3)^2\right)dx} \,.
\end{equation}
To show that the obtained operator is indeed one of those classified above, we consider a new change of variables
\begin{equation}
w^1=\frac{\bar{u}^1-\bar{u}^3}{\sqrt{2}},\qquad w^2=\bar{u}^2,\qquad w^3=\frac{\bar{u}^3-\bar{u}^1}{\sqrt{2}} \,.
\end{equation} 
The degenerate first order operator is written with the leading coefficient $\bar{g}=d\bar{u}^1\otimes d\bar{u}^1$ and the skew-symmetric bivector 
\begin{equation*}
    \bar{\omega}=-\sqrt{2}\bar{u}^3\left(d\bar{u}^1\wedge d\bar{u}^2 - d\bar{u}^2 \wedge d\bar{u}^3\right) \,.
\end{equation*}
The operator \eqref{1368} is of type $C^{\,i j}_{2,2}$ in three components showed in \eqref{eq:C2_three}. In particular, 

\vspace{1ex}

\centering{
    \begin{tabular}{l l}
       $g(v,w)=0\,,$  & $\hspace{2ex} f(v,w)=l(w)h(v,w)\,,$ \\[1.5ex]
        $l(w)=-1\,,$ & $\hspace{2ex} h(v,w)=\sqrt{2}w\,.$
    \end{tabular}
}
\end{example}

\vspace{1ex}

\begin{example}[Generalised KdV equation]\label{ex3} Let us consider the generalised KdV equation
\begin{equation}\label{mkdv}
    u_t+3(n+1)\,u^n\,u_x+u_{xxx}=0
\end{equation}
where $n$ is a positive integer. It is known that \eqref{mkdv} is Hamiltonian with the operator $\partial_x$ for any $n$. The case $n=2$ corresponds to the modified KdV equation (mKdV), it is integrable and it has a second Hamiltonian structure, with operator $\partial_x^3 + 6\, \partial_x \, u\, \partial_x^{-1}\, u\, \partial_x$. The Hamiltonians associated with mKdV are 
\begin{align}
    H_1 &= \int \left( \dfrac{3}{4} \, u^4 + \dfrac{1}{2}\,  u_x^2 \right) dx \,,  \hspace{10ex}    H_2 = \int \dfrac{1}{2} u^2 \, dx \,.
\end{align}

In \eqref{mkdv} we introduce the variables $u^1=u,\,u^2=u_{x},\,u^3=u_{xx}$, so that the equation reads as a quasilinear system of first order PDEs
\begin{equation}\label{sysmkdv}
%    \begin{cases}
    u^1_x=u^2\,, \qquad 
    u^2_x=u^3\,, \qquad 
    u^3_x=-u^1_t-3(n+1)(u^1)^n \, u^2 \,.
%    \end{cases} \,.
\end{equation}
The Hamiltonian structure is still conserved after the scalar equation is transformed into a system, i.e.\ \eqref{sysmkdv} has Hamiltonian structure with the operator 
\begin{equation}\label{opmkdv}
\begin{split}
    C^{\,ij}&=\begin{pmatrix}
    0&0&0\\0&0&0\\0&0&1
    \end{pmatrix}\,\partial_t+\begin{pmatrix}
    0&1&0\\-1&0&-3(n+1)(u^1)^{n-1}\\
    0&3(n+1)(u^1)^{n-1}&0
    \end{pmatrix}
\end{split} 
\end{equation} and the Hamiltonian functional 
\begin{equation}
    H=\int\left( 3(u^1)^{n+1}-u^1u^3+\frac{(u^2)^2}{2}\right) \, dx.
\end{equation}
The operator \eqref{opmkdv} is $C^{\,ij}_{3,2}$ in \eqref{eq:C2_three} with the exchange $u^1 \leftrightarrow u^3$ and

\vspace{2ex}

\centering{
    \begin{tabular}{l l}
       $g(u^1,u^2)=0\,,$  & $\hspace{2ex} f(u^1,u^2)=h(u^1,u^2)\,l(u^1)\,,$ \\[1.5ex]
        $l(u^1)=3(n+1)(u^1)^{n-1}\,,$ & $\hspace{2ex} h(u^1,u^2)=-1\,.$
    \end{tabular}
}
\end{example}

\vspace{2ex}

\begin{remark}
Let us observe that for $n>2$ the generalised KdV equation is not integrable, even if it is Hamiltonian as proved in the previous example. We emphasise that this feature is more general than the integrability property. 
\end{remark}

\vspace{1ex}

We finally present two examples violating the hypothesis of locality, either in terms of the momentum or of the operator defined for the inverted Hamiltonian structure.

\begin{example}
We consider the linearised KdV equation 
\begin{equation}
    u_t = u_{xxx} \,,
\end{equation}
for which the inverted system is easily given in the new variables by 
\begin{equation}
    u^1_x = u^2 \, \quad
    u^2_x = u^3 \, \quad
    u^3_x = u^1_t  \,.
\end{equation}
The associated momentum is given in terms of the  density
$  p(u) = \partial_x^{-2} u $. 
Here again, it is not possible to write the resulting system with Hamiltonian operator of type $1+0$ \cite{tsa3}. 
\end{example}
\begin{example}
We consider the Harry-Dym equation \cite{koinho,harry_dym_bi_hamiltonian}
\begin{equation}
\begin{split}
    u_t &= \left( \dfrac{1}{\sqrt{u}} \right)_{xxx} = -\dfrac{15}{8} u^{-7/2} \left(u_x\right)^3 + \dfrac{9}{4} u^{-5/2} u_x \, u_{xx} - \dfrac{1}{2} u^{-3/2} u_{xxx} \,,
\end{split}
\end{equation}
admitting the Hamiltonian structures 
\begin{align}
    u_t &= A_1 \, \dfrac{\delta H_1}{\delta u} = -\dfrac{1}{2} \, \partial_{x}^3 \,\dfrac{\delta H_1}{\delta u} \,, \qquad \text{ with } H_1= -\int 4\,\sqrt{u} \,\, dx \,, \\[2ex]
    u_t &= A_2 \, \dfrac{\delta H_2}{\delta u} = -\left( 2 u \partial_x - u_x \right) \dfrac{\delta H_2}{\delta u} \,, \qquad 
 \text{ with } H_2= - \int \left( \dfrac{15}{32} \, u^{-7/2}\,u_x - \dfrac{1}{16} \, u^{-5/2} \,u_{xx} \right) dx \,.
\end{align}
Introducing the variables $u_x = u^2, \, u_{xx} = u^3$, the inverted system is 
\begin{equation}
    %\begin{cases}
    u^1_x = u^2\,, \qquad 
    u^2_x = u^3 \,, \qquad 
    u^3_x = - 2 (u^1)^{3/2} u^1_t - \dfrac{15}{4} (u^1)^{-2} (u^2)^3 + \dfrac{9}{2} u^1\, u^2\, u^3 \,.
  %  \end{cases} \,.
\end{equation}
The momentum $P$ associated with the operator $A_2$ is 
\begin{equation*}
    u_x = - \left(  2 u \, \partial_x - u_x \right) \dfrac{\delta P}{\delta u} \,, \hspace{10ex} P = \int p(u) \, dx = - \int u\, dx \,,
\end{equation*}
and following the above procedure the Hamiltonian $H'$ as a functional in the new variables is 
\begin{equation}
    H' = - \int \left( \dfrac{3}{4} (u^1)^{-5/2} \, (u^2)^2 - \dfrac{1}{2} (u^1)^{-3/2} \, u^3 \right) \, dx \,.
\end{equation}
With this Hamiltonian it is not possible to build a local operator of the form $1+0$ for the inverted system, hence this operator will be non-local. 
\end{example}

%As a further perspective, it would be interesting to investigate the properties of the class of systems admitting this representation once inverted.  

\section{Conclusions}
The study of non-homogeneous quasilinear systems of first order PDEs is an ongoing research topic in integrable systems and Hamiltonian PDEs. 
To the authors' knowledge, a general criterion to discuss integrability for this kind of systems is not currently known, unlike the homogeneous systems~\cite{tsarev91:_hamil}. This paper represents a first step towards the investigation of integrability of non-homogeneous systems, focusing on the Hamiltonian property. The study of possible bi-Hamiltonian structures associated with these type will be the subject of a future paper. 
Non-homogeneous operators of order $k+m$ play an important role in nonlinear phenomena and their investigation represents another interesting topic \cite{LSV:bi_hamil_kdv,koinho,falq,JANNELLI2023107010,math10060935}.  Even in the simplest case where $k=1$, and $m=0$ we showed how the conditions for the operator to be Hamiltonian lead to a specific form, this being exactly solvable.  Higher order operators require a more general study, especially for what concerns the degenerate case. 

As future perspectives, the authors emphasise the necessity to further investigate the integrability of non-homogeneous quasilinear systems, the compatibility conditions between systems and operators in the sense of \cite{VerVit1}, but also their associated geometric structure, following the lead of the homogeneous case, where both operators and systems are %strictly 
linked to projective algebraic geometry \cite{AgaFer,FPV14,FerPavVit1,VerVit2} and differential Riemannian geometry.  Finally, the discrete analogous of non-homogeneous operators were introduced by Dubrovin in \cite{dubrovin89}, letting the classification method described in this paper suitable for discrete operators as well. 

\vspace{1ex} 

\textbf{Acknowledgements} The authors thank C. Benassi, F. Coppini, E. V. Ferapontov, A. Moro,  M. Pavlov and R. Vitolo for stimulating discussions and interesting remarks. 

PV  also acknowledges the financial support of GNFM of the Istituto Nazionale di Alta Matematica, of PRIN 2017 \textquotedblleft Multiscale phenomena in Continuum Mechanics: singular limits, off-equilibrium and transitions\textquotedblright, project number 2017YBKNCE and the research project Mathematical Methods in Non Linear Physics (MMNLP) by the Commissione Scientifica Nazionale -- Gruppo 4 -- Fisica Teorica of the Istituto Nazionale di Fisica Nucleare (INFN). 

MD would like to thank the Isaac Newton Institute for Mathematical Sciences, Cambridge, for support and hospitality during the programme ``Dispersive hydrodynamics: mathematics, simulation and experiments, with applications in nonlinear waves'' (HDY2) where work on this paper was undertaken. This work was supported by EPSRC grant no EP/R014604/1. 

\newpage

\bibliography{biblio}
\bibliographystyle{plain}

\newpage 
\appendix
{
\section{Derivation of the operators in Theorem \ref{thmde2} and Theorem \ref{thmde3}  
}

{
We give the details of the procedure followed to compute the classifications of Section \ref{3}. 
The computations are carried out with the support of computer algebra systems (Maple, Reduce and Mathematica) and finally checked by hand. 

For the sake of simplicity, we describe the first nontrivial operator obtained in three components, in the text identified as $C_{3,2}$. We start by considering the degenerate operator of order $1$ for a system in three components in the Savoldi's classification \cite{sav1} 
\begin{equation} \label{eq:g_b_simple}
    g^{ij}=\begin{pmatrix}
    \partial_x&0&0\\
    0&0&0\\
    0&0&0
    \end{pmatrix} , \qquad   b^{ij}_k=\begin{pmatrix}
    0&0&0\\
    0&0&0\\
    0&0&0
    \end{pmatrix} \,,  \qquad \forall k \in \{1,2,3\}\,.
\end{equation}
We add to the operator an order $0$ operator
\begin{equation}
\begin{pmatrix}\partial_x&0&0\\0&0&0\\0&0&0
    \end{pmatrix}+
    \begin{pmatrix}
        \omega^{11}&\omega^{12}&\omega^{13}\\\omega^{21}&\omega^{22}&\omega^{23}\\
        \omega^{31}&\omega^{32}&\omega^{33}
    \end{pmatrix} \,.
    \label{eq:omega_proof}
\end{equation}
Being $\omega$ an ultralocal tail, its elements $\omega^{ij}$ are functions at most depending on the three variables of the system $u,v,w$. 

The operator \eqref{eq:omega_proof} as a whole is Hamiltonian if its parts are Hamiltonian and the compatibility conditions established in Theorem II.3 are satisfied. In particular, the ultralocal term $\omega$ is Hamiltonian if it fulfills the Theorem II.1. We reduce the number of free functions in \eqref{eq:omega_proof} by using the skew-symmetry property 
\begin{equation}
    \begin{pmatrix}\partial_x&0&0\\0&0&0\\0&0&0
    \end{pmatrix}+
    \begin{pmatrix}
        0&\omega^{12}&\omega^{13}\\-\omega^{12}&0&\omega^{23}\\-\omega^{13}&-\omega^{23}&0
    \end{pmatrix} \,,
\end{equation}
so that the unknown functions are
\begin{equation} \label{eq:arbi}
    \omega^{12}(u,v,w), \quad \omega^{13}(u,v,w), \quad \omega^{23}(u,v,w) \,.
\end{equation}
To implement the constraints expressed in Theorem II.3, we evaluate the tensor $\Phi^{i j k}$ for the case \eqref{eq:g_b_simple}
\begin{align}
        \Phi^{i j k} & = g^{i s}\, \frac{\partial \omega^{j k}}{\partial u^s} - b^{i j}_{s} \, \omega^{s k} - b^{i k}_{s} \, \omega^{j s} = g^{i s}\, \frac{\partial \omega^{j k}}{\partial u^s} \,,
\end{align} 
where the sum over $s$ is intended via repeated indices. For three components $s=1,2,3$ and we use the notation for the variables $u^1 = u$, $u^2 = v$, $u^3=w$. The only non-zero element in $g$ is $g^{1 1}=1$, hence the constraint (26) on the non-zero elements of the tensor $\Phi^{ijk}$ takes the form $\Phi^{1jk}=\Phi^{k1j}$, with
\begin{equation}
\begin{split} 
    \Phi^{1 j k} &= \begin{pmatrix}
        0 & \dfrac{\partial \omega^{1 2}}{\partial u} & \dfrac{\partial \omega^{1 3}}{\partial u} \,\, \\[2.3ex] 
        -\dfrac{\partial \omega^{1 2}}{\partial u} & 0 & \dfrac{\partial \omega^{2 3}}{\partial u} \,\, \\[2.3ex]
        -\dfrac{\partial \omega^{1 3}}{\partial u} & -\dfrac{\partial \omega^{2 3}}{\partial u} & 0 
    \end{pmatrix}  = \begin{pmatrix}
        0 & \quad 0 & \quad 0 \qquad \\[3ex] 
        \dfrac{\partial \omega^{1 2}}{\partial u} & \quad 0 & \quad 0 \qquad \\[3ex] 
        \dfrac{\partial \omega^{1 3}}{\partial u} & \quad 0 & \quad 0 
    \end{pmatrix} = \Phi^{k1j} \,.
\end{split} 
\end{equation}
The constraints on the field variables are 
\begin{equation} \label{eq:arbi_1}
    \dfrac{\partial \omega^{ij}(u,v,w) }{\partial u} = 0 \quad \implies \quad  \omega^{ij}(u,v,w) =\omega^{ij}(v,w) \,,
\end{equation}
i.e.\ they do not depend on the variable $u$. We introduce the notation 
\begin{equation} \label{eq:notation}
    \begin{cases}
        \omega^{12}(v,w) = f(v,w)  \\ \omega^{13}(v,w)  = g(v,w)  \\ \omega^{23}(v,w)  = h(v,w) 
    \end{cases} \,.
\end{equation}
The constraint (27) yields
\begin{equation}
    \dfrac{\partial^2 \omega^{ij}}{\partial u^2} = 0 \,, \quad \dfrac{\partial^2 \omega^{ij}}{\partial u \, \partial v} = 0 \,, \quad 
    \dfrac{\partial^2 \omega^{ij}}{\partial u \, \partial w} = 0 \,,
\end{equation}
not producing any further restriction for the form of the functions, given \eqref{eq:arbi_1}. 
Finally, the closure requirement (19) is 
\begin{equation}
    \omega^{12} \, \dfrac{\partial \omega^{23}}{\partial v} -  \omega^{23} \,\dfrac{\partial \omega^{12}}{\partial v} +
 \omega^{13} \, \dfrac{\partial \omega^{23}}{\partial w}  - \omega^{23} \, \dfrac{\partial \omega^{13}}{\partial w} = 0 \,,
\end{equation}
with the notation \eqref{eq:notation} this becomes 
\begin{equation}
\begin{split}
    &f(v,w) \, \dfrac{\partial h(v,w)}{\partial v} -  h(v,w) \,\dfrac{\partial f(v,w)}{\partial v} +
 g(v,w) \, \dfrac{\partial h(v,w)}{\partial w}  - h(v,w) \, \dfrac{\partial g(v,w)}{\partial w} = 0 \,.
\end{split}
\end{equation}
We solve the last constraint with respect to the field $f(v,w)$. Observing that  
\begin{equation}
    \dfrac{\partial}{\partial v} \left(\dfrac{f}{h}\right) = \dfrac{1}{h} \dfrac{\partial f}{\partial v} - \dfrac{f}{h^2} \dfrac{\partial h}{\partial v} \,,
\end{equation}
we obtain the expression for $f(v,w)$ given in (33). The operator is then 
 \begin{equation*}
        C_{3,2}^{\,i j} = 
        \begin{pmatrix}
        \partial_x & f(v,w) & g(v,w) \\[1ex]
        -f(v,w) & 0 & h(v,w) \\[1ex]
        -g(v,w) & -h(v,w) & 0 \\
        \end{pmatrix} \,.
    \end{equation*}

%We can finally reproduce this procedure for all the other operators. The equations obtained by the hamiltonianity conditions can be explicitly solved in order to obtain the structure of $\omega^{ij}$.
 
 We remark that for fixed rank of $g^{ij}$ the resulting conditions  strongly depend on the structure of $b^{ij}_k$. For instance, for the operator $C_{3,3}^{ij}$ we have the same operator $g^{ij}$ as in \eqref{eq:g_b_simple}, but different $b^{ij}_k$
 \begin{equation} 
    b^{ij}_1=b^{ij}_2=\begin{pmatrix}
    0&0&0\\
    0&0&0\\
    0&0&0
    \end{pmatrix} , \qquad   b^{ij}_3=\begin{pmatrix}
    0&1&0\\
    -1&0&0\\
    0&0&0
    \end{pmatrix} \,.
\end{equation}
We look for the corresponding operator $\omega$ after considering the skew-symmetry property, hence we have three field variables $\omega^{1 2},\,\omega^{1 3},\,\omega^{2 3}$. 

\noindent 
The first conditions are imposed by comparing the tensors
 \begin{equation}
\begin{split} 
    \Phi^{1 j k} &= \begin{pmatrix}
        0 & -\omega^{1 3} + \dfrac{\partial \omega^{1 2}}{\partial u} & \dfrac{\partial \omega^{1 3}}{\partial u} \,\, \\[2.3ex] 
        \omega^{1 3} -\dfrac{\partial \omega^{1 2}}{\partial u} & 0 & \dfrac{\partial \omega^{2 3}}{\partial u} \,\, \\[2.3ex]
        -\dfrac{\partial \omega^{1 3}}{\partial u} & -\dfrac{\partial \omega^{2 3}}{\partial u} & 0 
    \end{pmatrix}  = \begin{pmatrix}
        0 & \quad 0 & \quad 0 \qquad \\[3ex] 
        -\omega^{1 3} + \dfrac{\partial \omega^{1 2}}{\partial u} & \quad \omega^{2 3} & \quad 0 \qquad \\[3ex] 
        \dfrac{\partial \omega^{1 3}}{\partial u} & \quad 0 & \quad 0   
    \end{pmatrix} =  \Phi^{k1j} \,.
\end{split} 
\end{equation}
At this stage we can already reduce the number of free functions, since $\omega^{2 3} = 0$. Solving the remaining equations, we find the dependence of the fields $\omega^{1 2}$, $\omega^{1 3}$ on the variables $u,\,v,\,w$
\begin{equation}
\begin{split}
    \omega^{1 2} (u,\,v,\,w) &= f(v,\,w) + u\,g(v,\,w) \\  \omega^{1 3}(u,\,v,\,w) &= g(v,\,w) \,.
\end{split}
\end{equation}
The constraints from the relation \eqref{cond4} further reduce the number of free functions, in particular $g(v,\,w) = 0$. The corresponding operator is then 
\begin{equation*}
      C_{3,3}^{\,i j} = 
        \begin{pmatrix}
        \partial_x & w_x + f(v,w) & 0 \\[1.5ex]
        -w_x -f(v,w) & 0 & 0 \\[1.5ex]
        0 & 0 & 0 \\
        \end{pmatrix} \,.
\end{equation*}
The same procedure has been carried out for all the possible forms of operators $g^{i j}$ and $b^{i j }_k$, obtaining the above mentioned classification.

}

}

\end{document}